\documentclass[18pt, reqno]{amsart}
\usepackage{amssymb,amsmath,mathrsfs,MnSymbol,setspace, xcolor}
\usepackage{setspace}
\usepackage{graphicx}
\usepackage{amsmath}
\usepackage{amsfonts}
\usepackage{epsfig}
\usepackage{xcolor}
\usepackage{url}
\usepackage{pgf}
\usepackage{tikz}
\usetikzlibrary{arrows,automata}

\makeatletter \oddsidemargin.9375in \evensidemargin \oddsidemargin
\def\authorsaddresses#1{\dedicatory{#1}}

\newcommand{\cal}{\mathcal}

\newtheorem{theorem}{Theorem}[section]

\newtheorem{proposition}[theorem]{Proposition}
\newtheorem{corollary}[theorem]{Corollary}
\theoremstyle{definition}
\newtheorem{definition}[theorem]{Definition}
\newtheorem{example}[theorem]{Example}

\theoremstyle{remark}
\newtheorem{remark}[theorem]{Remark}
\numberwithin{equation}{section}

\begin{document}
\setcounter{page}{1}


\title[ Some Improvements in Fuzzy Turing Machines]{  Some Improvements in Fuzzy Turing Machines}

\author[Hadi Farahani]{Hadi Farahani}


\authorsaddresses{Department of Computer Science, Shahid Beheshti University, G.C, Tehran, Iran\\
h$_-$farahani@sbu.ac.ir}


\begin{abstract}
In  this paper we modify some previous definitions of fuzzy Turing machines to  define the notions of accepting and rejecting degrees of inputs, \textit{computationally}. 
We use a  BFS-based search method and  obtain an upper level bound to guarantee the existence of accepting and rejecting degrees.
We show that fuzzy, generalized fuzzy and classical Turing machines have the same computational power.  
Next, we introduce the class of Extended Fuzzy Turing Machines equipped with  \textit{indeterminacy} states. 
These machines are used to catch some types of loops of  the classical Turing machines. Moreover, to each r.e. or co-r.e language, we correspond a fuzzy language which is indeterminable by an extended fuzzy Turing machine.

\end{abstract}
\maketitle




%


\section{Introduction}

In \cite{algorithm}, Lotfi Zadeh  defines the notion of fuzzy algorithm. His definition is based on a fuzzification of Turing machines. However, that work was not deep enough in the recursion theoretical aspects of the
mentioned model. That work is followed by the same setting in \cite{Lee}. The equivalency of previous fuzzy models is shown in \cite{Santos, Santos1}. Afterwards, the research in this field is revisited in \cite{Gerla87,Gerla1982,Harkteroad}. In \cite{Gerla87}, Biacino and Gerla generalize the definition of recursive enumerability introduced in \cite{Harkteroad}. 
%
%
Next,  Wiedermann  proposes a formal fuzzy computing model based on Turing machine model \cite{Wiedermann1, Wiedermann}. He claims that it is possible to accept r.e. sets and co-r.e. sets by fuzzy Turing machines and so these machines can solve the halting problem. So he claims that these fuzzy Turing machines have more computational power than the classical Turing machines.
In  \cite{2008}, Bedregal and Figueira  analyse Wiedermann's statement about the computational power of fuzzy Turing machines and show that Wiedermann's statement is not completely correct. They give a characterization of the class of n-r.e.  sets 
in terms of associated fuzzy languages recognized by fuzzy Turing machines. They also show that there is no universal fuzzy Turing machine which can simulate each machine  in the class of all fuzzy Turing machines. 
More recently, Moniri  defines the class of Generalized Fuzzy
Turing Machines \cite{Moniri}. His machines are equipped with both accepting and rejecting states. He studies some basic computability aspects of his machines and proves that a fuzzy language $L$ is decidable if and only if $L$ and $L^{c}$ are acceptable. 

In Section 3, we make some  essential modifications in previous definitions of fuzzy Turing machines in \cite{2008,Moniri,Wiedermann}. In these works, the notion of accepting or rejecting degree of an input is not \textit{computationally} well defined and there are some cases which these degrees \textit{can not be computed}. We modify these notions 
by (1) applying a BFS-based search  in the computational tree of a given fuzzy Turing machine on an input 
and (2) obtaining an upper bound on the number of levels in our BFS-based  search.   If we reach an accepting or a rejecting configuration in a level, then the upper bound indicates the number of next levels which  are needed to be traversed (at most) to \textit{determine the existence of another accepting or rejecting configuration}.
We also prove that  computational power of fuzzy, generalized fuzzy Turing machines, and classical machines are the same.

In Section 4, we establish a new class of  fuzzy Turing machines  that we call  Extended Fuzzy Turing Machines. These machines are equipped with 
\textit{indeterminacy} states as a new type of states. 
There are \textit{silent} transitions  between indeterminacy states and all other\textit{ non-accept} and \textit{non-reject} states with degree 1. By silent transition we mean a transition which does not change the position of head and the tape's content of the machine.
Indeterminacy states are applied to  catch some types of loops. The loop on an input is identified when its indeterminacy degree equals 0.  
Although extended fuzzy Turing machines are used to catch some types of loops, however they are not strong enough to solve the halting problem.
We  study some  basic computability properties of these machines. Moreover, 
to each r.e. or co-r.e language, we correspond a fuzzy language which is indeterminable by an appropriate extended  fuzzy Turing machine.

\section{ Preliminaries}
 In this section, we review some preliminaries from the  literature. 
 
 \begin{definition} \cite{Zadeh}
A fuzzy subset $A$ of a set $X$ is a function $\mu_{A}: X\rightarrow [0, 1]$, where for each $x\in X$,  $\mu_{A} (x)$  represents the grade of membership of the element $x\in X$ to $A$.
\end{definition}
 \begin{definition} \cite{Hajek}
A t-norm is a binary operation $\ast$ on $[0,1]$ satisfying commutativity, associativity, non-decreasing in both arguments, with the properties $0 \ast x = 0, 1 \ast x = x$, for all $x$. 
\end{definition}

\begin{definition}\cite{Moniri} Let $L_1$ and $L_2$ be two fuzzy languages. Also, let $\ast$ be a t-norm and $\ast^{'}$ be its dual t-conorm. The languages $L_{1}\ast L_{2}$ and $L_{1}\ast^{'}L_{2}$ are defined as follows:
$$( L_{1}\ast L_{2})(x)=L_{1}(x)\ast L_{2}(x) \ \ \qquad (L_{1}\ast^{'}L_{2})(x)=L_{1}(x)\ast^{'}L_{2}(x).$$
\end{definition}

%
%

\begin{definition} \cite{2008}
A  fuzzy Turing machine (FTM) is a triple ${\cal F}=({\cal T}, \ast, \mu)$, where 
${\cal T}=(Q, \Sigma, \Gamma, \Delta, q_s, F)$ is a non-deterministic Turing machine, $\ast$ is a t-norm and $\mu:Q \times \Gamma \times Q \times \Gamma \times \{R, L\} \rightarrow [0, 1]$ is a function.
In Turing machine $\cal T$, Q is a set of states, $\Sigma$ is a set of input symbols, $\Gamma$ is a set of tape symbols, transition relation $\Delta$ is a subset of $Q \times \Gamma \times Q \times \Gamma \times \{R, L\}$, $q_s \in Q$ is the starting state and $F \subseteq Q$ is the set of final states.
\end{definition}

\begin{definition}\cite{Moniri}
A Generalized fuzzy Turing machine (GFTM) is a tuple 
${\cal F}=({\cal T}, \ast, \ast^{'}, \eta)$, where ${\cal T}=(Q, \Sigma, \Gamma, \Delta, q_s, F)$ is a non-deterministic Turing machine such that  Q is a set of states, $\Sigma$ is a set of input symbols, $\Gamma$ is a set of tape symbols, $\Delta$ is a  \textbf{fuzzy subset} of $Q \times \Gamma \times Q \times \Gamma \times \{R, L\}$,  $q_s \in Q$ is the starting state, $F \subseteq Q$ is a set  of accepting and rejecting states, $\ast$ is a t-norm,  $\ast^{'}$ is the dual t-conorm of $\ast$ and $\eta : Q \times \Gamma \times Q \times \Gamma \times \{R, L\} \rightarrow [0, 1]$ is a function.
\end{definition}


\section{Corrections in Some Previous Definitions of FTMs }
In this section we  explain some computational objections about the definitions of fuzzy Turing machines in \cite{2008,Moniri,Wiedermann}. 
 The Wiedermann's idea in defining a Fuzzy Turing Machine (FTM) is to establish an uncertainty degree for the acceptance of a given string or the membership degree of a string to the language of a FTM. In order to compute this degree, he applies a composition on a t-norm evaluation. He considers 
the class of FTMs as a fuzzy extension of non-deterministic Turing machines (NTMs), where each transition has a membership degree. In order to define the acceptance degree of an input $w$, he defines the degree $d(C_t)$ of an arbitrary configuration $C_t$,  as the \textit{maximum} degree of all computational path $C_0 \preceq C_1 \preceq...\preceq C_t$ leading from the initial configuration $C_0=q_{0}w$ to $C_t$:
\begin{align*}
d(C_t)=max\{&degree(C_0,C_1,...,C_t)\in[0,1] :\\
&C_0 \preceq C_1\preceq...\preceq C_t ~\text{\textit{is a computational path from }$C_0$ \textit{to} $C_t$}\}
\end{align*}
He also defines the  accepting degree of a string $w$ in a FTM ${\cal F}=({\cal T}, \ast, \eta)$ as follows: 
$$deg_{\cal F}(w)=max\{d(uq_{f}v)\in[0,1] : q_{0}w \preceq^{*} uq_{f}v ~\text{for some $q_f\in F$}\}$$
where, ${\cal T}=(Q, \Sigma, \Gamma, \Delta, q_s, F)$ is a non-deterministic Turing machine and $F$ contains only \textit{accept} states.
Note that $q_{0}w \preceq^{*} uq_{f}v$ denotes that there is a computational path from $q_{0}w$ to $ uq_{f}v$.
 In \cite{2008,Moniri}, the same setting is followed. 
Bedregal and Figueira  \cite{2008}  change the \textit{maximum} to \textit{supremum} in  definitions of $d(C_t)$ and $deg_{\cal F}(w)$.  Moniri \cite{Moniri} assumes that $F$ may contain both \textit{accept} and  \textit{reject} states. Analogously, he defines accepting and rejecting degrees of a string $w$ using a composition on a t-conorm $\ast^{'}$ instead of \textit{maximum}, where $\ast^{'}$ is the dual of the t-norm $\ast$. 

For simplicity, we only continue discussion in the case of accepting degrees and our discussion can be extended to the case of rejecting degrees naturally. 
Here we use the notation $deg_{\cal F}(w)$ for accepting degree of $w$.

The important point is that \textit{none of the definitions above for  accepting degree of an input 
is not computationally well defined.} 
There are some cases in which these degrees can not be computed.  
The machine ${\cal T}$ is non-deterministic and when the machine is running on an input $w$ there are different computational branches. 
If there is a loop branch, then the machine continues searching in computational tree forever to find all accepting configurations. 
 This prevents the process to halt.
%
%
%
%
For example, the  machine illustrated in Figure 1, has infinitely many computational paths leading from $q_{s}01$ to $0q_{I}1$ and does not stop searching in the computational tree on input $01$. 
Thus, degree of configuration $0q_{I}1$ can not be \textit{computationally} obtained using  definitions in \cite{2008,Moniri,Wiedermann}.

\begin{figure}[t]
\centering \includegraphics[scale=0.4]{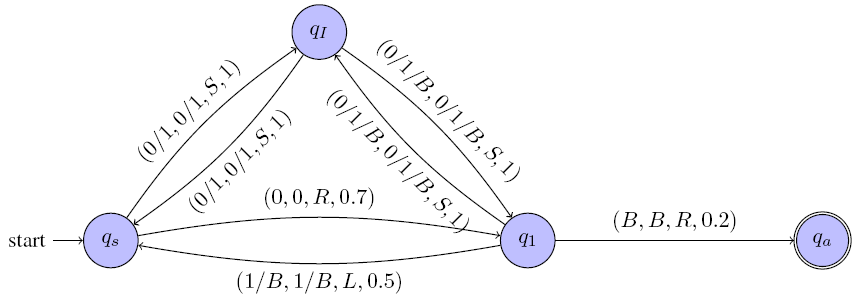}
\caption{ An Extended Fuzzy Turing Machine }
\label{pic2}
\end{figure}

In order to resolve this problem, we make some changes in  previous definitions. We take Moniri's  definition of $deg_{\cal F}(w)$ as the base and  make/use the following changes/facts:
\begin{itemize}
\item[$\bullet$] Any non-deterministic Turing machine can be simulated by a deterministic Turing machine  using a BFS search method (See Theorem 3.16 \cite{Sipser}). 
By applying and using some modifications in the introduced BFS search method, we simulate the non-deterministic Turing machine $\cal T$ by a deterministic Turing machine.
%
\item[$\bullet$] We find an upper bound for the number of levels (in our BFS-based search) which are needed to be traversed (at most) to guarantee the existence of $deg_{\cal F}(w)$,
\item[$\bullet$] We define  path independent degree  just for \textit{accepting} and \textit{rejecting} configurations. 
\item[$\bullet$] We assume that the function  $\eta : Q \times \Gamma \times Q \times \Gamma \times \{R, L\} \rightarrow [0, 1]$, 
gives  values 0 and 1  in the following special cases: 
\begin{itemize}
\item[(i)] a direct transition from initial state to an accept or a reject state is allowed to take degree 0,
\item[(2)]  transitions \textit{into} or \textit{out from} an indeterminacy state are allowed to take degree 1.
\end{itemize}
%
%
\end{itemize}
%
 Suppose that ${\cal F}=({\cal T}, \ast, \eta)$ is an FTM, where ${\cal T}=(Q, \Sigma, \Gamma, \Delta, q_s, F)$ is a non-deterministic Turing machine. 
  Let  $\cal D$ be the deterministic Turing machine which simulates  $\cal T$ 
  using the  BFS search method described in Theorem 3.16 \cite{Sipser}.
  We make the following modifications in the introduced BFS search method:
  \begin{itemize}
  \item[(i)]  If in BFS search an accepting configuration in a node  is reached, 
then we don't traverse any of the following  configurations initiated from that node. This modification is important because   for example if we have a silent move from  accept state to itself, i.e. for some $a\in \Gamma$ we have $(q_{accept}, a, q_{accept}, a, S)\in \Delta$, then we have a branch consisting infinitely many accepting configurations which are not effective in computing  $deg_{\cal F}(w)$.
    \item[(ii)]  Assume that  we are at level $n$ when $\cal D$ is simulating  $\cal T$ on  an input $w$.
   If we  visit an accepting configuration in this level, then we can find an upper bound for the number of levels  enough to be traversed (at most) to find another accepting configuration. If we could not find another accepting configuration in this  level bound, then it is guaranteed that there is no other accepting configuration and we can halt and compute $deg_{\cal F}(w)$. 
  \end{itemize}

In the sequel, we find the upper's level bound. Suppose that we reach an accepting configuration $C$ with fuzzy degree $d$ in level $n$.
 Consider our modified BFS in traversing the computational tree on input $w$. If we visit no new accepting configuration in next  $t=max_{d^{'}}\lceil log_{k}d/d^{'}\rceil$ levels, then we are sure that there is no other accepting configuration and we can compute $deg_{\cal F}(w)$. In this upper  bound, $maximum$ is taken over all degrees $d^{'}$ of  \textit{another configurations in level $n$ }and $k$ is the maximum value in the range of $\eta$. Note that $\Delta$ is a finite set and so the image of $\eta$ over $\Delta$ is a finite set and has a maximum value $k$.
%
%

Now we explain the idea to construct the given upper bound.
 In the definition of  the degree of a given path, the t-norm $\ast$ is applied between degrees of transitions along the path.
%
%
%
  So, if the degree of transitions is strictly less than 1, then the path's degree strictly decreases  when we traverse from the current level to next one in the computational tree.
  Assume that the current configuration in level $n$ 
  is $C^{'}$ with degree $d^{'}$ and $i$ is the minimum natural number which satisfies the condition $d^{'}\ast k^{i}\leq d$. Thus, for configurations in an arbitrary path initiated from $C^{'}$,  the only ones effective in computing $deg_{\cal F}(w)$ are (at most) in the next $``i"$ levels.
  %
%
  If we apply the same process for each arbitrary configuration in level $n$, then the upper bound $t=max_{d^{'}}\lceil log_{k}d/d^{'}\rceil$ can be defined as the maximum of these numbers ``$i$". 
%
%
%

Consequently, by the discussion above,  we can modify  previous definitions of FTMs or GFTMs to  compute the accepting or rejecting degrees of inputs in general. In Section 4, we apply the modifications so far in defining the class of Extended Fuzzy Turing Machines.

\begin{remark}
Let $\cal E=({\cal T}, \ast, \eta)$ be an FTM and the computational tree of ${\cal T}$ on an input $w$ has \textit{no loop branch}, then the accepting degree of each arbitrary configuration can be computationally defined and also there is no restriction on degrees of transitions. A similar discussion  also holds for GFTMs.
\end{remark}
\begin{remark}
If in an FTM $\cal E=({\cal T}, \ast, \eta)$ (or similarly in a GFTM), the machine  ${\cal T}$ be a \textit{deterministic} Turing machine, then all accepting or rejecting degrees in \cite{2008,Moniri,Wiedermann} are computationally well-defined.
\end{remark}
%
%
Now we compare the  computational power of fuzzy and generalized fuzzy Turing machines to classical Turing machines. Concerning Church-Turing thesis, this comparison is an  important issue in computability theory.
 It is known that FTMs and GFTMs are extensions of NTMs. In the following it is shown that there is a classical Turing machine which can simulate each FTM (also, in the same setting for GFTMs) and so, FTMs, GFTMs and classical Turing machines have the same computational power.

\begin{theorem}
There exists a classical Turing machine which can simulate each fuzzy Turing machine  on an input $w$ and yields the accepting  degree of $w$.
\end{theorem}
\begin{proof}
We give the sketch of the proof. Let ${\cal F}=({\cal T}, \ast, \eta)$ be an FTM, where, ${\cal T}=(Q, \Sigma, \Gamma, \Delta, q_s, F)$ is a non-deterministic Turing machine. First note that the language of ${\cal F}$ includes pairs of the form $(w, deg_{\cal F}(w))$. 
Assume that $D$ is the deterministic Turing machine which simulates $\cal T$ by applying our modified BFS search. We propose a classical 3-tapes deterministic Turing machine $M$ which can simulate the machine  ${\cal F}$ such that for each input $w$:
\begin{itemize}
\item[$\bullet$] $M$ accepts $w$ iff   ${\cal T}$ accepts $w$,
\item[$\bullet$] If ${\cal T}$ accepts $w$, then $M$ outputs $deg_{\cal F}(w)$.
\end{itemize} 
Intuitively the machine $M$ works as follows:
\begin{itemize}
\item[$\bullet$] Its first tape is the input tape which remains unchanged during the computation,
\item[$\bullet$] Its second tape is the work tape and the simulation of $\cal T$  on $w$  is executed on this tape,
\item[$\bullet$] Its third tape holds both the degree of the current configuration and  also the degrees of  visited accepting configurations.
\end{itemize}
If after visiting all accepting configurations (considering the upper level bound) the machine $D$ halts on $w$,  then the machine $M$ computes the degree of $w$ using  the contents of its third tape and outputs this degree.
\end{proof}
Likewise, it can be shown that there exists a classical Turing machine which can simulate each GFTM and give the accepting or rejecting degree of a given input. 

\section{Extended  Fuzzy Turing Machines and Some Computability Results}

In this section, 
we extend FTMs and  GFTMs  to  machines with some new type of  states that we call \textit{indeterminacy} states. We also  study some computability properties of these extended machines. 
%

\subsection{Extended  Fuzzy Turing Machines}

Considering the modifications mentioned in Section 3, we define   the class of Extended  Fuzzy Turing Machines as an extension of the class of Generalized Fuzzy Turing Machines (GFTMs) defined by Moniri \cite{Moniri}. In defining extended  fuzzy Turing machines, we consider some modifications with respect to  Moniri's definition: (1) we specify \textit{indeterminacy} states as a new type of states; 
(2) unlike Moniri, we assume that the transition relation $\Delta$ is a \textit{\textbf{crisp}} set and it is not \textit{\textbf{fuzzy}}.
%


\begin{definition}\label{Definition4.1}
 An \textit{Extended  Fuzzy Turing Machine} (briefly, EFTM) is a tuple 
${\cal E}=({\cal T}, \ast, \ast^{'}, \eta)$, where ${\cal T}=(Q, \Sigma, \Gamma, \Delta, q_s, F )$ is a non-deterministic Turing machine. $Q$ is a set of states which consists a special set $\cal I$ of \textit{ indeterminacy} states, $\Sigma$ is a set of input symbols, $\Gamma$ is a set of tape symbols containing  the blank symbol, $\Delta$ is a \textit{\textbf{crisp subset}} of $Q \times \Gamma \times Q \times \Gamma \times \{R, L,S\}$, ${q_{s}\in Q}$ is the start state, $F={\cal A}\cup{\cal R}\subseteq Q$ consists accepting and rejecting states,
$\ast$ is a t-norm and $\ast^{'}$ is its dual t-conorm and $\eta : Q \times \Gamma \times Q \times \Gamma \times \{R, L,S\} \rightarrow [0, 1]$ is a function which corresponds  a \textit{truth degree} to each move in $\Delta$ such that: 
\begin{itemize}
\item[$(i)$] transitions that lead to an indeterminacy state or exit from it, do not change the head position and the tape's content, i.e. for each transition $(q,a,p,b,m)\in Q \times \Gamma \times Q \times \Gamma \times \{R, L,S\}$:
$$q\in \cal I ~\text{or}~p\in \cal I~~~\Longrightarrow   \ \ \  a=b \text{ and } m=S,$$ 
\item[$(ii)$] for each $(q,a,p,b,m)\in Q \times \Gamma \times Q \times \Gamma \times \{R, L,S\}$:
$$\eta(q,a,p,b,m)=1 ~~~\text{ iff   \ \ \ at least  one of } q \text{ or } p \text{ is in } \cal I,$$ 
\item[$(iii)$] a direct transition from  initial state to an accept or a reject state can take degree 0, i.e. for each $(q,a,p,b,m)\in Q \times \Gamma \times Q \times \Gamma \times \{R, L,S\}$:
$$\eta(q,a,p,b,m)=0 ~~~\text{ iff }~~~ q=q_s \text{ and } (p=q_{a}  \text{ or } p=q_{r}) $$ 
where, $q_{a}$ and $q_{r}$ are accept  and reject states, respectively.
\end{itemize}

\end{definition}

\begin{remark}
Note that although  the transition relation $\Delta$ in a GFTM  
is considered as a \textit{fuzzy subset}, but we  define $\Delta$ as a \textit{crisp subset} in our model. In this way,  $\eta$ has the desired meaning.
\end{remark}
%
Instantaneous description (ID) which is the unique description of a machine's tape, is defined as usual (see \cite{Wiedermann1}). If $\alpha,\alpha^{'}$ are two IDs, then $\alpha\preceq^{r}\alpha^{'}$ means that there is a move
in $\Delta$ with truth degree $r$ leading from $\alpha$ to $\alpha^{'}$ in one step. 
%
%
Let $\alpha_{t}$ be reachable from $\alpha_{0}$ in $t$ steps through the computational path $\alpha_{0} \preceq^{r_{1}} \alpha_{1} \preceq^{r_{2}} ... \preceq^{r_{t}} \alpha_{t}$, then the truth degree of this path is defined as $r_{1}\ast...\ast r_{t}$, if  $ t\geq1$ and 1, if $t=0$.
Due to non-determinism, an arbitrary configuration such as $\alpha_t$ can be reached from $\alpha_0$ through different computational paths, and so the degree of a configuration should be defined path independently. 
%
We define the path independent degree of a configuration just for \textit{accepting} and \textit{rejecting} configurations. Let $\alpha_h$ be an accepting or a rejecting configuration, then we define its degree as follows:
$$d(\alpha_h)=\ast^{'}_{j\in J}a_{j}$$
 where, $\{a_j: j\in J\}$ is the set of all truth degrees of all computational paths leading from the initial ID $\alpha_0$ to $\alpha_h$ such that the computational paths are traversed considering the obtained upper level bound for the given modified BFS search. 
%

We call the computational path $\alpha_{0} \preceq^{r_{1}} \alpha_{1} \preceq^{r_{2}} ... \preceq^{r_{t}} \alpha_{t}$,
an accepting, a rejecting or an indeterminacy path, if $\alpha_{t}$ is an accepting, a rejecting or an indeterminacy ID, i.e. in the form $uq_a v$, or $uq_r v$ or $uq_i v$, where $u$, $v$ are two strings in $\Gamma^{\ast}$ and $q_a$, $q_r$, $q_i$ are  states in ${\cal {A, R, I}}$, respectively.
%

\begin{definition} \label{inf}
 Let ${\cal E}=({\cal T}, \ast, \ast^{'}, \eta)$ be an EFTM and  $w\in \Sigma^{*}_{\cal T}$. If there exists at least one accepting (or rejecting) path on input $w$, then the accepting (or rejecting)  degree of $w$, denoted by  $e_{{\cal E}}(w)$ (or $e_{{\cal E}}^{'}(w))$, is defined as follows:
$$e_{{\cal E}}(w)=max_{\alpha_{a}}\{d(\alpha_a)~ |~ q_s w\preceq^{\ast} \alpha_{a}\}$$
$$(or~e_{{\cal E}}^{'}(w)=max_{\alpha_{r}}\{d(\alpha_r)~ |~ q_s w\preceq^{\ast} \alpha_{r}\})$$
%
%
where, $\alpha_{a}$ (or $\alpha_{r}$) is an accepting (or a rejecting) ID. Otherwise, if there is no desired path we define
$e_{{\cal E}}(w)=0(=e_{{\cal E}}^{'}(w))$. 
\end{definition}
\begin{remark}
In the definition of an EFTM, the set of states $Q$ and the set of symbols $\Gamma$ are  finite and so the sets of  accepting and rejecting  configurations on a given input are finite. Therefore, unlike Moniri's, in Definition \ref{inf} we use  \textit{maximum} instead of supremum in  the definitions of accepting and rejecting  degrees.
\end{remark}


%
%

Above, we computationally defined  path independent truth degree for an accepting or a rejecting configuration only. Now, we  define the path independent truth degree of an\textit{ indeterminacy }configuration $\alpha_i$  such that it is not computationally definable in general. 
Assume that ${\cal E}=({\cal T}, \ast, \ast^{'}, \eta)$ is an EFTM and  $w\in \Sigma^{\ast}_{\cal T}$. Let $\alpha_i$ be an indeterminacy configuration on $w$, then we define:
$$d(\alpha_i)=Inf_{j\in I}b_{j}$$
where, for each $j$,   $b_j$ is the  truth degree of a computational path leading from $\alpha_0$ to $\alpha_i$. 
Note that $d(\alpha_i)=Inf_{j\in I}a_{j}$ is mathematically definable. 
If there exists at least one indeterminacy path on an  input $w$, then we define the indeterminacy degree of $w$ as follows:
$$e_{{\cal E}}^{''}(w)=min_{\alpha_{i}}\{d(\alpha_i)~ |~ q_s w\preceq^{\ast} \alpha_{i}\}.$$
where, $\alpha_{i}$  is an indeterminacy ID. Otherwise, $e_{{\cal E}}^{''}(w)=0$.

\begin{example}
Consider the EFTM shown in Figure 1. It can be verified that the accepting degree of  $w_{1}=0$ equals 0.14, and the indeterminacy degree of  $w_{2}=01$ equals 0.
%
%
\end{example}

%
%

In the sequel, we give a proposition which characterizes the loops of  classical Turing machines using  EFTMs.
\begin{proposition}\label{power}
For each   classical Turing machine $M$ that loops on an input  $w$, 
 there exists an EFTM $\cal K$ such that  $e_{\cal K}^{''}(w)=0$.
\end{proposition}
\begin{proof} The machine $\cal K$ is constructed from $M$ as follows:
\begin{itemize}
\item[$\bullet$] consider the start state of $M$ as its start state,
\item[$\bullet$] correspond a non-trivial degree in open interval (0,1)  to each transition of $M$,
\item[$\bullet$] add a state $q_I$ as  indeterminacy state and consider  transitions between $q_I$ and  all non-accept and non-reject states of $M$.
\end{itemize}
 It can be shown that if the machine $M$ loops on an input $w$, then the indeterminacy degree of $w$ is 0, i.e. $e_{\cal K}^{''}(w)=0$. 
\end{proof}

\begin{remark}\label{power1}
By Proposition \ref{power},  EFTMs can catch the loops but, these machines are not powerful enough to catch some non-halting cases such as \textit{configuration expansion} and so EFTMs can not solve the halting problem.
\end{remark}

%

Let ${\cal E}=({\cal T}, \ast, \ast^{'}, \eta)$ be an EFTM. We define $A({\cal E})$, $R({\cal E})$ and $I({\cal E})$, as  fuzzy languages accepted,  rejected or  indeterminated by $\cal E$, which their membership functions are  $e_{{\cal E}}$, $e_{{\cal E}}^{'}$ and $e_{{\cal E}}^{''}$, respectively. In this setting we can think of ${\cal E}$ as the 
triple $({e_{{\cal E}}},e_{{\cal E}}^{'},e_{{\cal E}}^{''})$. Note that $I({\cal E})$ contains all pairs $(w,e_{{\cal E}}^{''}(w))$ which their first element  $w$ is neither accepted nor rejected by ${\cal T}$. So, intuitively $I({\cal E})$ consists all pairs $(w,e_{{\cal E}}^{''}(w))$  which ${\cal T}$ does not halt  on $w$.

 We  define the notion of an  \textit{acceptable} or a \textit{decidable} fuzzy language  similar to \cite{Moniri}, but we define the new notion of an \textit{indeterminable} fuzzy language in the following definition. 
\begin{definition}\label{definition}Let $L$ be a fuzzy language.$ $ \\
\begin{itemize}
\item[$(i)$] $L$ is acceptable if there is an EFTM ${\cal E}$ such that $L =A({\cal E})$,\\
\item[$(ii)$] $L$ is decidable if there is an EFTM ${\cal E}$ such that $L = A({\cal E})$ and $L^{c}=R({\cal E})$,\\
\item[$(iii)$] $L$ is indeterminable if there is an EFTM ${\cal E}$ such that $L=A({\cal E})$ and $L < I({\cal E})$, i.e. for each $w\in L$ we have $e_{{\cal E}}(w)<e_{{\cal E}}^{''}(w)$.

%
\end{itemize}
\end{definition}
 Note that $L^{c}$ is defined as $L^{c}(x)=1-L(x)$. In the following we restate Proposition 3.5, Corollary 3.6 and  Proposition 3.8 of \cite{Moniri}, which also hold here.

\begin{proposition}
A fuzzy language $L$ is decidable if and only if $ L$ and $L^c$ are acceptable.
\end{proposition}

\begin{corollary} \label{decidable}
 If a fuzzy language $L$ is decidable, then $L^c$ is decidable,
\end{corollary}

\begin{proposition} 
Let $L_1$ and $L_2$ be two fuzzy languages. Assume that $L_1$ and $L_2$ are accepted by EFTMs ${\cal E}_1$ and ${\cal E}_2$ equipped with the same t-norm $\ast$ and let $*^{'}$ be the dual t-conorm of $\ast$. Then $L_{1}\ast^{'} L_{2}$ is accepted by an EFTM equipped with the same t-norm and t-conorm.
\end{proposition}


\subsection{Classical Languages and Extended Fuzzy Turing Machines}

In this section, we propose a correspondence between the class of extended fuzzy Turing machines and the class of all classical r.e. or co-r.e. languages. 
 
 \begin{proposition}\label{1}
 Let $L$ be a classical r.e. language. There is an EFTM $\cal E$ such that for each arbitrary input $w$, we have $w \in L$ if and only if ${\cal E}$ accepts $w$ with a non-zero degree $b$, rejects it with degree 0 and indeterminates it with a degree strictly greater than $b$. 
 \end{proposition}
\begin{proof}
Suppose that $M$ is a classical Turing machine which recognizes $L$. Without loss of generality we  assume that $M$  has only accepting states as its final states and eventually accepts an input $w$ if $w$ is in $L$ and never halts on other inputs. Let $t$ be a real number in (0,1). Construct the EFTM ${\cal E}$ as follows:
\begin{itemize}
\item[$\bullet$] change the starting state of $M$ to a non-starting state $q$,
\item[$\bullet$] correspond degree $t$ to all transitions of $M$,
\item[$\bullet$] consider a starting state $q_s$,
\item[$\bullet$] consider two nondeterministic transitions from $q_s$ to:\\
 $~~~~~~~~(i)$  state $q$ of the new modified version of $M$ with degree $t$, \\
 $~~~~~~~~(ii)$  a new rejecting state with degree 0 (note that here the degree 0 is allowed here by Definition \ref{Definition4.1}),
\item[$\bullet$]  consider transitions from all non-accept and non-reject states to  an  indeterminacy state with degree 1 and vice versa.
\end{itemize}
It can be shown that for each input $w$ if $w\in L$,  then there is a number $b\in (0,1)$ such that ${e_{\cal E}}(w)=b\leq t$ and ${e^{''}_{\cal E}}(w)> b$, also $ {e^{'}_{\cal E}}(w)=0$. Therefore, $\cal E$ is the desired machine.

\end{proof}
Remind that in above proposition, the language  $L$ is \textit{ indeterminable} by $\cal E$.

\begin{proposition}\label{2}
Let $L$ be a (classical)  co-r.e language. There is an EFTM ${\cal E}$ such that for each input $w$, we have $w \notin L$ if and only if ${\cal E}$ rejects $w$ with a non-zero degree $r$, accepts it with degree 0 and indeterminates it with a degree strictly greater that $r$.
\end{proposition}
\begin{proof}
Since $L$ is a co-r.e language then there exists a Turing machine $M$ which recognizes $L^c$. Let $s$ be a real number in (0,1). Construct the EFTM ${\cal E}$ as follows: 
\begin{itemize}
\item[$\bullet$]  replace the starting state of $M^′$ with a new non-starting state $q$,
\item[$\bullet$] correspond degree $s$ to each transition of $M^′$,
\item[$\bullet$]  change the accepting states of $M^′$ to rejecting states,
\item[$\bullet$] consider a starting state $q_s$,
\item[$\bullet$] add  two non-deterministic transitions from the start state $q_s$: \\$~~~~~~~$ (1)  to the state $q$ with degree $s$, 
\\$~~~~~~~$ (2)  to an accepting state with degree 0 (note that here the degree 0 is allowed here by Definition \ref{Definition4.1}), 
\item[$\bullet$] consider transitions from all non-accept and non-reject states to  an  indeterminacy state with degree 1 and vice versa.
\end{itemize}

It can be shown that for each input $w$ if $w\in L^{c}$,  then there is a number $r\in (0,1)$ such that ${e^{'}_{\cal E}}(w)=r\leq s$ and ${e^{''}_{\cal E}}(w)> r$, also $ {e_{\cal E}}(w)=0$. Thus, $\cal E$ is the desired machine.

\end{proof}
%

\section{Final remarks}


We gave some modifications in previous definitions of fuzzy Turing machines proposed in \cite{2008,Moniri,Wiedermann}.  
We applied a  BFS-based search method, also obtained an upper level bound to define  the notions of  accepting  and rejecting degrees of a given input, computationally.
We introduced the class of Extended Fuzzy Turing Machines 
which are equipped with  \textit{indeterminacy} states. Finally, we used indeterminacy states to catch the loops of classical Turing machines. 



\begin{small}


\end{small}

%
%

\end{document}